\documentclass[a4paper,UKenglish,cleveref, autoref, thm-restate, nolineno]{socg-lipics-v2021}
\usepackage{xspace,enumitem,mathtools}

\pdfoutput=1 
\hideLIPIcs  

\title{Density Approximation for Moving Groups} 

\author{Max van Mulken}{TU Eindhoven, the Netherlands}{m.j.m.v.mulken@tue.nl}{https://orcid.org/0000-0001-6609-2057}{}

\author{Bettina Speckmann}{TU Eindhoven, the Netherlands}{b.speckmann@tue.nl}{https://orcid.org/0000-0002-8514-7858}{}

\author{Kevin Verbeek}{TU Eindhoven, the Netherlands}{k.a.b.verbeek@tue.nl}{https://orcid.org/0000-0003-3052-4844}{}

\authorrunning{M. van Mulken, B. Speckmann and K. Verbeek} 

\Copyright{Max van Mulken, Bettina Speckmann, Kevin Verbeek} 



\category{} 

\relatedversion{} 

\nolinenumbers 

\EventEditors{John Q. Open and Joan R. Access}
\EventNoEds{2}
\EventLongTitle{42nd Conference on Very Important Topics (CVIT 2016)}
\EventShortTitle{CVIT 2016}
\EventAcronym{CVIT}
\EventYear{2016}
\EventDate{December 24--27, 2016}
\EventLocation{Little Whinging, United Kingdom}
\EventLogo{}
\SeriesVolume{42}
\ArticleNo{23}

\usepackage{nccmath}
\usepackage{hyperref}
\usepackage[nocompress]{cite}

\newcommand{\etal}{\textit{et al.}\xspace}

\DeclareMathOperator{\poly}{poly}

\newcommand{\KDE}{\text{KDE}}
\newcommand{\WT}{\widetilde{T}}
\newcommand{\Reals}{\mathbb{R}}
\newcommand{\dom}{\mathcal{D}}

\newcommand{\es}{\varepsilon_{\text{dsc}}}

\newcommand{\ec}{\varepsilon_{\text{cor}}}

\newcommand{\SN}{\mathcal{M}}

\begin{document}

\maketitle

\begin{abstract}
Sets of moving entities can form groups which travel together for significant amounts of time. Tracking such groups is an important analysis task in a variety of areas, such as wildlife ecology, urban transport, or sports analysis. Correspondingly, recent years have seen a multitude of algorithms to identify and track meaningful groups in sets of moving entities. However, not only the mere existence of one or more groups is an important fact to discover; in many application areas the actual shape of the group carries meaning as well. In this paper we initiate the algorithmic study of the shape of a moving group. We use kernel density estimation to model the density within a group and show how to efficiently maintain an approximation of this density description over time. Furthermore, we track persistent maxima which give a meaningful first idea of the time-varying shape of the group. By combining several approximation techniques, we obtain a kinetic data structure that can approximately track persistent maxima efficiently.
\end{abstract}

\section{Introduction}\label{sec:introduction}

Devices that track the movement of humans, animals, or inanimate objects are ubiquitous and produce significant amounts of data. Naturally this wealth of data has given rise to a large number of analysis tools and techniques which aim to extract patterns, and ultimately knowledge from said data. One of the most important patterns formed by both humans and animals are groups: sets of moving entities which travel together for a significant amount of time. Identifying and tracking groups is an important task in a variety of research areas, such as wildlife ecology, urban transport, or sports analysis. Consequently, in recent years various definitions and corresponding detection and tracking algorithms have been proposed, such as  herds~\cite{herds1}, mobile groups~\cite{groups2005}, clusters~\cite{clusters2005, gennady2010cluster}, and flocks~\cite{reynolds1987flocks, benkert2008reporting}. 

In computational geometry, there is a sequence of papers on variants of the \emph{trajectory grouping structure} which allows a compact representation of all groups within a set of moving entities~\cite{grouping-structure, trajectorygroupinggeodesic, van2016grouping, isaac-KreveldLSW16, gis-WiratmaKLS19}. We follow the same definitions and notation with a size parameter $m$, a temporal parameter $\delta$, and a spatial parameter $\sigma$. A set of entities forms a group during time interval $I$ if it consists of at least $m$ entities, $I$ is of length at least $\delta$, and the union of the discs of radius $\sigma$ centered around all entities forms a single connected component.

However, not only the mere existence of one or more groups is an important fact to discover, in many applications areas it is equally important to detect how individuals within a group or the group as a whole are moving. Several papers hence focus on defining~\cite{wood2011detecting}, detecting~\cite{gudmundsson2008movement}, and categorizing~\cite{wood2011detecting, movementpatterns} movement patterns of complete groups. These patterns are based on coordinated behavior of the individuals within the group; for example, a possible pattern is a group of animals which all exhibit foraging behavior. Another type of such patterns are formations, such as geese flying in a V-shape or groups following a leader~\cite{laube2005discovering, grouptransformations}. Naturally, in this context there are also several papers which focus on detecting roles in sports teams~\cite{lucey2014get, lucey2013team} and identifying formations in football teams~\cite{sportsanalysis, playingstyle, formationdetection, GUDMUNDSSON201416}.

Apart from the actions of the individual moving entities in a group, the actual shape of the group and the density distribution also carry meaning. Consider, for example, the herd of wildebeest in Fig.~\ref{fig:groups}. Both its global shape and the distribution of dense areas indicate that this herd is migrating. Research in wildlife ecology~\cite{jasperphd, densityvsprey} has established that animals often stay close together when not under threat and respond to immediate danger by spreading out. Hence from the density and the extent of a herd we can infer fear levels and external disturbances.
The density distribution and general shape of a group are not only meaningful in wildlife ecology, but they can also provide useful insights when monitoring, for example, visitors of a festival to detect the onset of a panic.
In this paper we initiate the algorithmic study of the shape of a moving group. Specifically, we identify and track particularly dense areas which provide a meaningful first idea of the time-varying shape of the group.  

\begin{figure}[t]
    \centering
    \includegraphics[page=2]{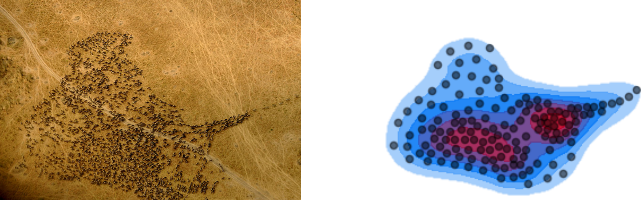}
    \caption{A herd of wildebeest, the shape of the group indicates migratory behavior.\footnotemark} 
    \label{fig:groups}
\end{figure}
\footnotetext{Photo by T. R. Shankar Raman, \href{https://commons.wikimedia.org/wiki/File:Wbeest_Mara.jpg}{https://commons.wikimedia.org/wiki/File:Wbeest\_Mara.jpg} on 25/11/2022.}

In this paper we initiate the algorithmic study of the shape of a moving group. Specifically, we identify and track particularly dense areas which provide a meaningful first idea of the time-varying shape of the group. It is our goal to develop a solid theoretical foundation which will eventually form the basis for a software system that can track group shapes in real time. We believe that algorithm engineering efforts are best rooted in as complete as possible an understanding of the theoretical tractability of the problem at hand. To develop an efficient algorithmic pipeline, we are hence making several simplifying assumptions on the trajectories of the moving objects (known ahead of time, piecewise linear, within a bounding box). In the discussion in Section~\ref{sec:discussion} we (briefly) explain how we intend to build further on our theoretical results to lift these restrictions, trading theoretical guarantees for efficiency.

\subparagraph{Problem statement.} Our input consists of a set $P$ of $n$ moving points in $\Reals^2$; we assume that the points follow piece-wise linear motion and are contained in a bounding box $\mathcal{D} = [0,D] \times [0,D]$ with size parameter $D$. The position of a point $p \in P$ over time is described by a function $p(t)$, where $t \in [0, T]$ is the time parameter. We omit the dependence on $t$ when it is clear from the context, and simply denote the position of a point as $p$. We refer to the $x$- and $y$-coordinates of a point $p$ as $x(p)$ and $y(p)$, respectively. 

We assume that the set $P$ continuously forms a single group; we aim to monitor the density of $P$ over time. We measure the density of $P$ at position $(x, y)$ using the well-known concept of \emph{kernel density estimation} (KDE)~\cite{KDE}. KDE uses \emph{kernels} around each point $p \in P$ to construct a function $\KDE_P\colon \Reals^2 \rightarrow \Reals^+$ such that $\KDE_P(x, y)$ estimates the density of $P$ at position $(x, y)$.  We are mainly interested in how the density peaks of $P$, that is, the local maxima of the function $\KDE_P$, change over time. Not all local maxima are equally relevant (some are minor ``bumps'' caused by noise), so we would like to only track the significant local maxima through time. We measure the significance of a local maximum using the concept of \emph{topological persistence}. 

\subparagraph{Approach and organization.} We could simply attempt to maintain the entire function $\KDE_P$ over time and additionally keep track of its local maxima. However, doing so would be computationally expensive, as the complexity of $\KDE_P$ is  at least quadratic in $n$ (in general, it grows exponentially with the dimension $d$ of its domain). Furthermore, there are good reasons to approximate $\KDE_P$: (1) KDE is itself also an approximation of the density, and (2) approximating $\KDE_P$ may directly eliminate local maxima that are not relevant. 

The following is a simple approach to (roughly) approximate $\KDE_P$: we build a quadtree on $P$, and we consider smaller cells to have higher density. This approach has the advantage that the spatial resolution near the density peaks is higher. However, there are two main drawbacks: (1) the approximation does not depend on the chosen kernel size for the KDE (this is an important parameter for analysis), and (2) we cannot guarantee that all significant local maxima of $\KDE_P$ are preserved in the approximation.

The approach we present in this paper builds on the simple approach described above, but eliminates its drawbacks. We first introduce the approximation of a function $f\colon \Reals^2 \rightarrow \Reals^+$ using a volume-based quadtree. Specifically, instead of subdividing a cell in the quadtree when it contains more than one point, we subdivide a cell in the quadtree if the volume under the function $f$ contained within the cell exceeds some pre-specified threshold. Furthermore, we assign a single function value to each leaf cell of the quadtree corresponding to the average value of $f$ within the cell. This results in a 2-dimensional step function that approximates $f$ (see Fig.~\ref{fig:example-approx}). In Section~\ref{sec:volume-quadtree} we discuss the volume-based quadtree in more detail and prove several properties concerning its complexity, structure, and how well it approximates the original function $f$, which may be of independent interest. 

\begin{figure}
    \centering
    \includegraphics[width=.7\textwidth]{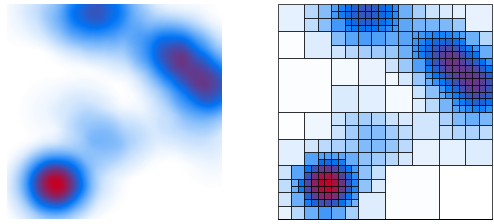}
    \caption{A 2-dimensional function (left) approximated as a step function (right).}
    \label{fig:example-approx}
\end{figure}

We aim to maintain a volume-based quadtree for $\KDE_P$ over time, without explicitly maintaining $\KDE_P$. In Section~\ref{sec:disc} we therefore show how to replace $\KDE_P$ by a set of points $Q$ that approximates the volume under $\KDE_P$. Specifically, we use the concepts of \emph{corests} and \emph{$\varepsilon$-approximations} to compute a small set of points that can accurately approximate the volume under $\KDE_P$. We choose the points in $Q$ such that we do not have to update $Q$ as long as the original points in $P$ do not change their trajectories. We can, however, update $Q$ efficiently when a trajectory changes. 

We now transformed the problem of maintaining local maxima of $\KDE_P$ to the problem of maintaining a quadtree on a set of moving points. In Section~\ref{sec:kinetic} we present a simple \emph{kinetic data structure} (KDS) that efficiently maintains the volume-based quadtree that approximates $\KDE_P$, as well as its local maxima, which correspond to the high persistence local maxima of $\KDE_P$. Theorem~\ref{thm:kds-final} formally states our result; $\poly(x)$ denotes a polynomial function in $x$.

\begin{theorem}\label{thm:kds-final}
Let $f = \KDE_P$ be a KDE function on a set $P$ of $n$ linearly moving points in $\Reals^2$. For any $\varepsilon > 0$, there exists a KDS that approximately maintains the local maxima of $f$ with persistence at least $2\varepsilon$. The KDS can be initialized in $O\left(n \poly\left(\frac{\log n}{\varepsilon}\right)\right)$ time, processes at most $O\left(n \poly\left(\frac{1}{\varepsilon}\right)\right)$ events, and can handle events and flight plan updates in $O\left(\log n + \poly\left(\frac{1}{\varepsilon}\right)\right)$ and $O\left(\poly\left(\frac{\log n}{\varepsilon}\right)\right)$ time, respectively.   
\end{theorem}

The individual concepts and techniques that we use are fairly straightforward, but the combination of all (KDE, quadtrees, coresets/$\varepsilon$-approximations, KDS, and topological persistence) to achieve a single goal is, to the best of our knowledge, quite unique, and certainly non-trivial. Our main contribution is therefore the proper combination of the various parts and the careful analysis that this requires. We briefly reflect on our approach (and its shortcomings) in Section~\ref{sec:discussion}.


\section{Preliminaries}\label{sec:prelim}

\subparagraph{Kernel density estimation.} Let $P$ be a set of points in $\Reals^2$. To obtain a kernel density estimation (KDE)~\cite{KDE, rosenblatt1956remarks} $\KDE_P$ of $P$, we need to choose a kernel function $K\colon \Reals^2 \rightarrow \Reals^+$ that captures the influence of a single point on the density. Let $\sigma$ denote the \emph{kernel width} and $\|(x, y)\|$ the Euclidean norm of $(x, y)$.
Some examples of common kernels include:\\
\begin{minipage}{.7\textwidth}
  \begin{description}
\item[Uniform.] $K(x, y) = \begin{cases}
1, & \text{for } \|(x, y)\| < \sigma \\
0, & \text{otherwise}
\end{cases}$\item[Cone.] $K(x, y) = \begin{cases}
1-\|(x, y)\|/\sigma, & \text{for } \|(x, y)\| < \sigma \\
0, & \text{otherwise}
\end{cases}$
\item[Pyramid.] $K(x, y) = \begin{cases}
1-\max(|x|, |y|)/\sigma, & \text{for } |x|, |y| < \sigma \\
0, & \text{otherwise}
\end{cases}$
\item[Gaussian.] $K(x, y) = \frac{1}{2 \pi \sigma^2} e^{-\frac{x^2 + y^2}{2 \sigma^2}}$
\end{description}
\end{minipage}
\begin{minipage}{.2\textwidth}
  \centering
  \includegraphics[width=.7\textwidth]{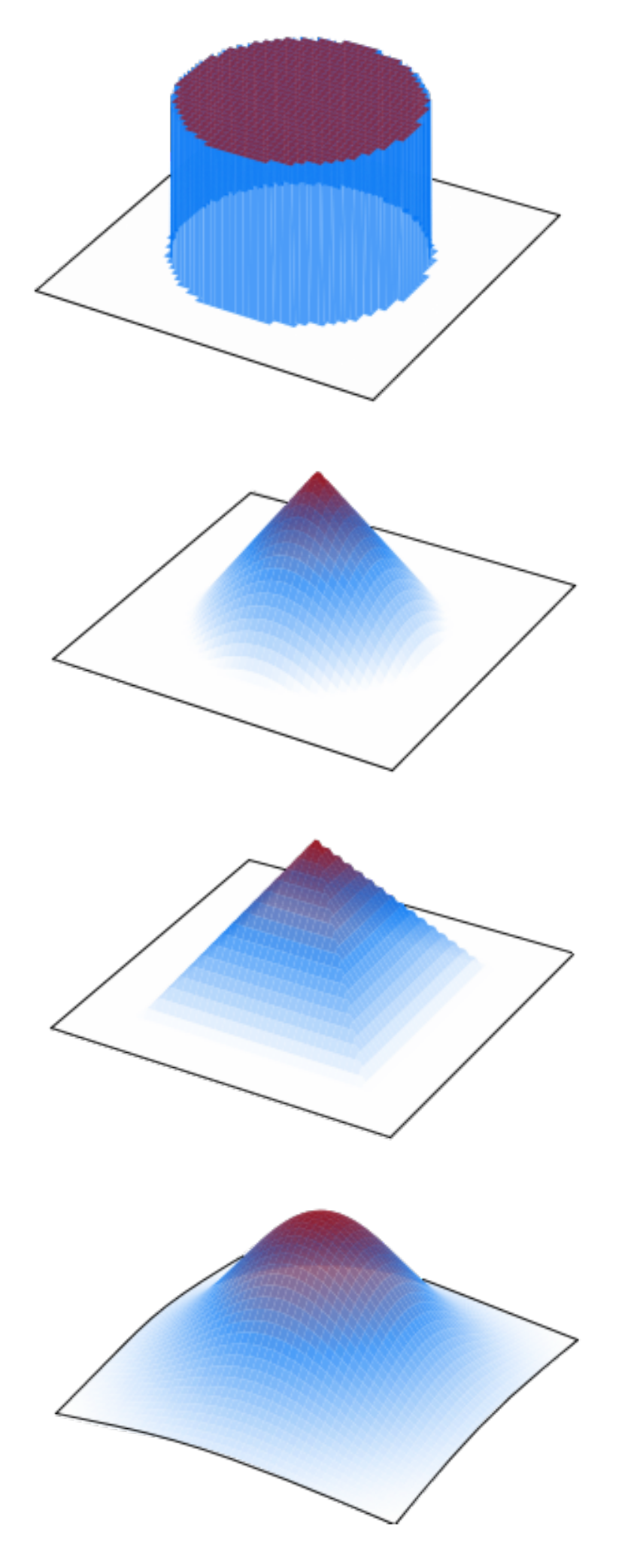}
\end{minipage}

In Section~\ref{sec:volume-quadtree} we show that our approach works with any kernel that has bounded slope. This excludes the uniform kernel. Furthermore, for simplicity we assume that the (square) range in which the kernel function produces nonzero values is bounded by kernel width $\sigma$. Kernel functions such as the Gaussian kernel do not have this property a priori, but can be easily adapted.
We also scale the input, without loss of generality, such that $\sigma = 1$, and the kernel function such that the volume under the kernel function is always $1$. We also assume that the maximum value attained by the kernel function is at most $1$. 
This property holds for most kernels and for all kernels we consider in this paper.

Given a kernel function $K$ and a point set $P$, we compute the KDE as follows:
\begin{equation}
\KDE_P(x, y) = \frac{1}{n} \sum_{p \in P} K(x - x(p), y - y(p)) \qquad \text{for } (x, y) \in \Reals^2.
\end{equation}

Note that the volume under $\KDE_P(x, y)$ is the same as the volume under a single kernel function, which is $1$. Also, the maximum value attained by $\KDE_P(x, y)$ is at most the maximum value attained by a single kernel function, which we assume to be at most $1$.

\subparagraph{Quadtrees.} Consider a bounded domain of size $D$ in $\Reals^2$ which is defined by the square $\dom = [0, D] \times [0, D]$. A \emph{quadtree} $T$ on this domain is a tree where each node $v \in T$ represents a region $R(v) \in \dom$. The region of the root $r \in T$ is the whole domain, $R(r) = \dom$.
Every node $v \in T$ is either a leaf, or it has exactly $4$ children. 
We create the regions of the children by partitioning the region $R(v)$ into $4$ equal-size regions along the vertical and horizontal line through the center of $R(v)$. The leaves of $T$ partition the domain $\dom$. We refer to a leaf in a quadtree $T$ also as a \emph{cell}. For a node $v \in T$ we further denote the side length of $R(v)$ as $s(v)$. Note that, if $w$ is a child of $v$ in $T$, then $s(w) = s(v)/2$.  Sometimes we need to consider leaf nodes that are neighbors to a leaf $v \in T$ in a spatial sense. We denote this set of neighbors of a leaf $v \in T$ by $\mathcal{N}(v)$ and a leaf $u \in T$ is in $\mathcal{N}(v)$ if and only if the closed regions $R(u)$ and $R(v)$ share a (piece of a) boundary, that is, $R(u) \cap R(V) \neq \emptyset$. See Fig.~\ref{fig:quadtree-example} for an example.

\begin{figure}
    \centering
    \includegraphics{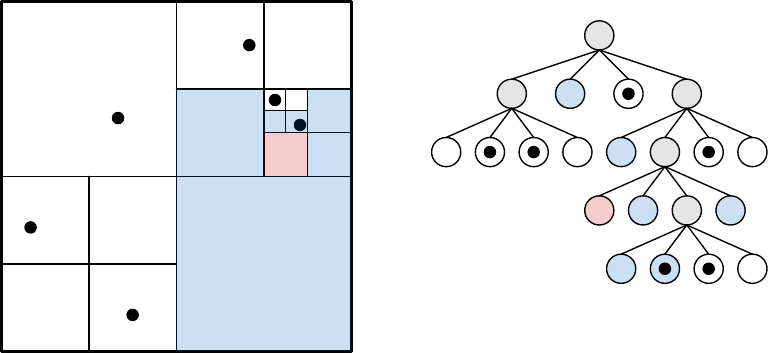}
    \caption{A quadtree on a small point set. The leaf nodes of the quadtree correspond to a set of regions that partition the domain. The neighborhood $\mathcal{N}(v)$ of the red node is shown in blue.}
    \label{fig:quadtree-example}
\end{figure}

Generally, a quadtree $T$ is constructed on a point set $P$; we subdivide a node $v \in T$ as long as $R(v)$ contains more than one point (or as long as the number of points exceeds some threshold). The depth of a quadtree may be unbounded in general, however, in our setting the depth is naturally bounded.
We use a quadtree $T$ to represent a function, and therefore we augment the quadtree by adding a value $h(v) \in \Reals^+$ for each leaf node $v \in T$. The corresponding function $f_T\colon \dom \rightarrow \Reals^+$ is then defined as the function for which $f_T(x, y) = h(v)$ if $(x, y) \in R(v)$ for some leaf $v \in T$. Note that $f_T$ is ill-defined on the boundary between cells, but for our approach the function is allowed to take the value $h(v)$ of any of the adjacent cells, and hence we will simply ignore the function values at the boundaries of cells.

\subparagraph{Coresets.} In the context of geometric approximation algorithms, a \emph{coreset} $Q$ of a data set for some algorithm $\mathcal{A}$ is a reduced data set such that running $\mathcal{A}$ on the coreset $Q$ provides an approximation for running the algorithm on the full data set. In this paper we are mostly interested in a specific type of coreset used to estimate the density of geometric objects, namely \emph{$\varepsilon$-approximations}. To define $\varepsilon$-approximations, we first need to introduce the concept of range spaces. A \emph{range space} $(S, \mathcal{R})$ consists of a finite set of objects $S$ and a set of ranges $\mathcal{R}$, where $\mathcal{R}$ is a set of subsets of $S$, that is, $\mathcal{R} \subseteq 2^{S}$. Typically, in the context of geometry, $S$ is a set of points, and $\mathcal{R}$ consists of all subsets of $S$ that can be covered by some geometric shape (e.g., a disk or square). Given a range space $(S, \mathcal{R})$, an \emph{$\varepsilon$-approximation} for some $\varepsilon > 0$ is a subset $Q \subset S$ such that, for all ranges $R \in \mathcal{R}$ we have that $\left|\frac{|R \cap Q|}{|Q|} - \frac{|R|}{|S|}\right| \leq \varepsilon$. 

An important concept in the context of range spaces is the VC-dimension. The \emph{VC-dimension} $d$ of a range space $(S, \mathcal{R})$ is the size of the largest subset $Y \subseteq S$ such that $\mathcal{R}$ restricted to $Y$ contains all subsets of $Y$, that is, $2^Y \subseteq \mathcal{R}_{|Y}$, where $\mathcal{R}_{|Y} = \{R \cap Y \mid R \in \mathcal{R}\}$. We say that $Y$ is \emph{shattered} by $\mathcal{R}$.
We use the following result by Agarwal~\etal\cite{approx-rangespace}.

\begin{restatable}[\hspace{-1sp}\cite{approx-rangespace}]{theorem}{dynamicapprox}
\label{thm:eps-approx-rangespace}
    Given a range space $X = (S, \mathcal{R})$ of VC-dimension $d$ and a parameter $\varepsilon$, one can (deterministically) maintain an $\varepsilon$-approximation of $X$ of size $O(\frac{1}{\varepsilon^2}\log(\frac{1}{\varepsilon}))$, in $O(\frac{\log^{2d+3}n}{\varepsilon^{2d+2}}(\log(\log(n)/\varepsilon))^{2d+2})$ time per insertion and deletion.
\end{restatable}

Although it is possible to prove better bounds on the size of an $\varepsilon$-approximation for certain specific range spaces, we use this result since it is generic, and it also allows for efficient insertion and deletion of elements, which is relevant for our approach. 

\subparagraph{Kinetic data structures.} Kinetic Data Structures (KDS) were introduced by Basch~\cite{KDSs} to efficiently track attributes of time-varying geometric objects, such as the convex hull of a set of moving points. A KDS uses so-called certificates to ensure that the attribute in question is unchanged. That is, the maintained attribute changes only when a certificate fails. Certificates are geometric expressions which are parameterized by the trajectories of the objects. In a classic KDS we assume that these trajectories, also called flight plans, are known at all times. We refer to the failure of a certificate as an event; all certificates are sorted by failure time and stored in an event queue. The KDS then proceeds to handle events one by one, updating its certificates and possibly also the tracked attribute.

The efficiency of a classic KDS is evaluated according to four quality criteria.
The \emph{responsiveness} of a KDS measures the worst-case amount of time necessary to restore the structure after an event; a KDS is responsive if each event can be handled in polylogarithmic time in the input size. The \emph{locality} of a KDS is determined by the maximum number of certificates that depend on the same object. The \emph{compactness} of a KDS measures the  maximum number of certificates that can exist at any one time. Finally, the \emph{efficiency} of a KDS is determined by the ratio between all events it has to handle versus the number of \emph{external events} which correspond to actual changes in the tracked attribute. An efficient KDS does not need many additional certificates, and hence events, beyond those strictly necessary to maintain the attribute.

\subparagraph{Topological persistence.}
Let $f\colon \Reals^2 \rightarrow \Reals$ be a (smooth) 2d function. A \emph{critical point} of $f$ is a point $(x, y)$ such that the gradient of $f$ at $(x, y)$ is $(0, 0)$. Critical points capture the overall (topological) structure of the function $f$, and are hence important for analysis. Critical points of a 2d function come in three types: \emph{local minima}, \emph{saddle points}, and \emph{local maxima}. While some critical points capture relevant features of the function $f$, other critical points may simply exist only due to noise. We can measure such relevance of a critical point via the concept of \emph{topological persistence}. 

Let $L_f(z) = \{(x, y) \mid f(x, y) \leq z\}$ be the sublevel set of $f$ with respect to $z \in \Reals$. Now consider the connected components of $L_f(z)$ as we increase the value of $z$. If $(x, y)$ is a local minimum of $f$, then a new connected component is formed in $L_f(z)$ at $z = f(x, y)$, and we call $(x, y)$ the \emph{representative} of that connected component. If $(x, y)$ is a saddle point of $f$, then at $z= f(x, y)$ either two connected components of $L_f(z)$ are merged, or a new loop/hole is created in $L_f(z)$. In the first case we pair $(x, y)$ with the representative $(x', y')$ of the connected component that was formed last. This pair of critical points is called a \emph{persistence pair}.
Finally, if $(x, y)$ is a local maximum of $f$, then a hole in $L_f(z)$ is closed at $z=f(x,y)$. We then pair $(x, y)$ with the most recently introduced saddle point $(x', y')$ responsible for creating the respective hole (avoiding exact definitions). 

If $((x, y), (x', y'))$ is a persistence pair of $f$, then we call $|f(x, y) - f(x', y')|$ the \emph{persistence} of both critical points $(x, y)$ and $(x', y')$ of $f$. Critical points with low persistence mostly arise due to noise, while critical points with high persistence are important for the overall structure of $f$. Therefore we are interested in local maxima with sufficiently high persistence. Now let $\hat{f}$ be an approximation of $f$, where the goal is to preserve the local maxima with high persistence in $\hat{f}$. 

Before investigating which maxima are maintained, we require a number of concepts. For a function $f$ we can create a \emph{persistence diagram} $D(f)$ by adding a point $(f(x, y), f(x', y'))$ for every persistence pair $((x, y), (x', y'))$ of $f$. Additionally, we add all points on the diagonal (that is, of the form $(z, z)$ for $z \in \Reals$) to $D(f)$. Given two functions $f$ and $g$, let $\|f-g\|_\infty = max_{(x,y) \in \Reals^2} |f(x, y) - g(x, y)|$ indicate the $L_\infty$ distance between $f$ and $g$. For two (multi)sets of points $X$ and $Y$ we can define the bottleneck distance as 
\[
d_B(X, Y) = \inf_\gamma \sup_{x \in X} \|x - \gamma(x)\|_\infty, 
\]
where $\gamma$ ranges over all bijections from $X$ to $Y$. A generalization of the following theorem is shown by Cohen-Steiner~\etal in~\cite{persistence-diagram-stability}.

\begin{theorem}[{\cite{persistence-diagram-stability}}]\label{thm:persistence}
Given two functions $f, g\colon \mathbb{R}^2 \rightarrow \mathbb{R}$, the persistence diagrams satisfy $d_B(D(f), D(g)) \leq \|f - g\|_\infty$.
\end{theorem}

With these definitions out of the way, we can prove the following lemma:

\begin{restatable}{lemma}{persistence}\label{lem:persistent-maxima}
    Given two functions $f, \hat{f}\colon \Reals^2 \rightarrow \Reals$ such that $|f(x,y) - \hat{f}(x,y)| < \varepsilon$ for all $(x,y) \in \Reals^2$, there exists an injection from the local maxima of $f$ with persistence $> 2\varepsilon$ to the local maxima of $\hat{f}$.
\end{restatable}
\begin{proof}
Consider the functions $f$ and $\hat{f}$. Since $|f(x, y) - \hat{f}(x, y)| < \varepsilon$ for all $(x, y) \in \Reals^2$, we have that $\|f - \hat{f}\|_\infty \leq \varepsilon$. By Theorem~\ref{thm:persistence} this implies that $d_B(D(f), D(\hat{f})) \leq \varepsilon$. Now consider a local maximum of $f$ with persistence at least $2 \varepsilon$. This critical point must hence be involved in a persistence pair $((x,y), (x', y'))$ such that $f(x', y') - f(x, y) \geq 2 \varepsilon$. This persistence pair corresponds to a point $(z, z')$ in the persistence diagram $D(f)$ with $z + 2 \varepsilon \leq z'$. Since $d_B(D(f), D(\hat{f}) \leq \varepsilon$, there must exist a point $(\hat{z}, \hat{z}') \in D(\hat{f})$ such that $\|(z, z') - (\hat{z}, \hat{z}')\|_\infty \leq \varepsilon$. As $z + 2 \varepsilon \leq z'$, the point $(\hat{z}, \hat{z}')$ cannot lie on the diagonal and must hence correspond to a persistence pair of $\hat{f}$. Thus, we can say that the local maxima of $f$ of persistence $>2\varepsilon$ can be mapped to local maxima in $\hat{f}$.
\end{proof}

\section{Volume-based quadtree}\label{sec:volume-quadtree}
In this section we analyze the approximation of a continuous, two-dimensional function $f$ by a quadtree $T$. Specifically, we prove certain properties on the structure of $T$ and how well the corresponding function $f_T$ approximates $f$.
We are given a function $f\colon \dom \rightarrow \Reals^+$ on a bounded square domain $\dom = [0, D] \times [0, D]$ of size $D$. Without loss of generality, we assume that the total volume under the function $f$ (over the whole domain $\dom$) is exactly $1$ (this can be achieved by scaling the domain/function). We construct the quadtree $T$ using a threshold value $\rho > 0$. Specifically, starting from the root, we subdivide a node $v \in T$ when the volume under $f$ restricted to the region $R(v)$ exceeds $\rho$, and we recursively apply this rule to all newly created child nodes. As a result, for every cell (leaf) $v \in T$, the volume under $f$ restricted to $R(v)$ is at most $\rho$. Finally, for every cell $v \in T$, we set the value $h(v)$ to the average of $f$ in $R(v)$, which implies that the volume under $f$ and $f_T$ is equal when restricted to $R(v)$. For convenience, we will denote this volume as $V(v) = s(v)^2 h(v)$. We will refer to the quadtree $T$ constructed in this manner as the \emph{volume-based quadtree} of $f$.


Given some constant $\varepsilon > 0$, our goal is to choose $\rho$ such that $\|f - f_T\|_\infty < \varepsilon$, that is, $|f(x, y) - f_T(x, y)| < \varepsilon$ for all $(x, y) \in \dom$. First of all, note that this is not possible for all possible functions $f$: if the slope of $f$ can be unbounded, then it is easy to construct an example where $f(x, y)$ and $f(x+\epsilon, y)$ are arbitrarily far apart, but $(x, y)$ and $(x+\epsilon, y)$ must belong to the same quadtree cell for any threshold $\rho$, assuming $\epsilon$ is chosen sufficiently small. We therefore express our bounds in terms of the Lipschitz constant\footnote{The Lipschitz constant of a function $f$ is the maximum absolute slope of $f$ in any direction.} $\lambda$ of $f$ and the maximum value $z^*$ of $f$, next to the size of the domain $D$ and the threshold value $\rho$.

We start by proving some simple properties on the complexity of $T$.

\begin{restatable}{lemma}{cellsize}\label{lem:cellsize}
    Let $T$ be the volume-based quadtree of $f\colon \dom \rightarrow \Reals^+$ with threshold $\rho$. Then, for any cell $v \in T$, we have that $s(v) > \frac{1}{2}\sqrt{\rho/z^*}$, where $z^* = \max_{(x, y) \in \dom} f(x, y)$.
\end{restatable}

\begin{proof}
    Assume for the sake of contradiction that there is a cell $v \in T$ with $s(v) \leq \frac{1}{2}\sqrt{\rho/z^*}$, and let $w$ be the parent of $v$ in $T$. By construction of $T$, the volume under $f$ when restricted to $R(w)$ must exceed $\rho$. As $s(w) \leq \sqrt{\rho/z^*}$, this implies that the average value of $f$ in $R(w)$ is more than $\rho / s(w)^2 = z^*$, which is clearly a contradiction. 
\end{proof}

\begin{corollary}\label{cor:quaddepth}
Let $T$ be the volume-based quadtree of $f\colon [0, D]^2 \rightarrow \Reals^+$ with threshold $\rho$. Then the depth of $T$ is at most $\log\left(\frac{2 D}{\sqrt{\rho/z^*}}\right)$, where $z^* = \max_{(x, y) \in [0, D]^2} f(x, y)$.
\end{corollary}

\begin{restatable}{lemma}{treecomplexity}\label{lem:treecomplexity}
Let $T$ be the volume-based quadtree of $f\colon [0, D]^2 \rightarrow \Reals^+$ with threshold $\rho$. Then $T$ has $O\left(\frac{1}{\rho} \log\left(\frac{D}{\sqrt{\rho/z^*}}\right)\right)$ nodes in total, where $z^* = \max_{(x, y) \in [0, D]^2} f(x, y)$.
\end{restatable}

\begin{proof}
Let $T'$ be the quadtree obtained by removing all leaves from $T$. By construction of $T$, the volume under $f$ restricted to $R(w)$ for any leaf node $w \in T'$ must exceed $\rho$. As the regions corresponding to all leaf nodes of $T'$ must be interior disjoint, $T'$ can have at most $\frac{1}{\rho}$ leaf nodes (recall that the total volume under $f$ is assumed to be $1$). Using Corollary~\ref{cor:quaddepth} we can then directly conclude that $T'$ contains at most $\frac{1}{\rho} \log\left(\frac{D}{\sqrt{\rho/z^*}}\right)$ nodes in total. As every node can have at most $4$ children, $|T| \leq 4 |T'|$ and the result follows.      
\end{proof}

Next, we investigate how well $f_T$ approximates the function $f$. To establish meaningful bounds, we first need bounds that relate the function values of $f$ to the volume under $f$. 

\begin{lemma}\label{lem:vol-bounds}
Let $f\colon \dom \rightarrow \Reals^+$ be a 2-dimensional function with Lipschitz constant $\lambda$, and let $R$ be a square region in $\dom$ with side length $s$. If $f(x, y) = z$ for some $(x, y) \in R$, then the volume $V$ under $f$ restricted to $R$ satisfies:
\begin{enumerate}[label=(\alph*)]
    \item If $z < \sqrt{2} \lambda s$, then $V \geq \frac{z^3}{6 \lambda^2}$,
    \item If $z \geq \sqrt{2} \lambda s$, then $V \geq s^2 (z - \frac{2\sqrt{2}}{3} \lambda s)$, 
    \item $V \leq s^2 (z + \frac{2\sqrt{2}}{3} \lambda s)$. 
\end{enumerate}
\end{lemma}

\begin{proof}
We start with proving (a) and (b). Consider any coordinates $(x, y) \in R$ and let $z = f(x, y)$. Imagine $f$ as a surface plotted in $\Reals^3$. By definition of the Lipschitz constant, we can construct a cone below $(x, y, z)$ with slope $\lambda$ such that everything inside that cone (and above the ground plane $z = 0$) must belong to the volume under the function $f$. Now consider the volume of the part of the cone that is restricted to $R$. Since $(x,y) \in R$, this volume is minimized when $(x, y)$ lies on one of the corners of $R$, so assume that this is the case. We now approximate this cone by a pyramid with the same apex that is contained completely within the cone (see Fig.~\ref{fig:cone}). Note that this pyramid has slope $\lambda$ along the diagonals, as that is where the slope of a pyramid is minimal. Again consider the volume under this pyramid restricted to $R$. We consider two cases. If $z < \sqrt{2} \lambda s$, then the boundary of the pyramid at the ground plane is still contained within $R$ (see Fig.~\ref{fig:pyramid-in-cell}). The volume is then given by $\frac{1}{3} (\frac{z}{\sqrt{2} \lambda})^2 z = \frac{z^3}{6 \lambda^2}$. Otherwise, the volume within $R$ consists of a box with a (quarter) pyramid on top. Then the volume is $\frac{1}{3} s^2 \sqrt{2} \lambda s + (z - \sqrt{2} \lambda s) s^2 = s^2 (z - \frac{2\sqrt{2}}{3} \lambda s)$. This concludes the bounds in (a) and (b).

\begin{figure}[t]
    \centering
    \includegraphics{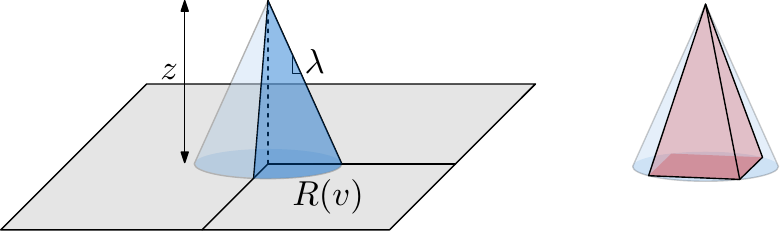}
    \caption{A cone with its apex at $(x, y, z)$ and slope $\lambda$ must be fully underneath $f$. The volume in the cone is lower-bounded by the volume of the pyramid in red.}
    \label{fig:cone}
\end{figure}

\begin{figure}[t]
    \centering
    \includegraphics{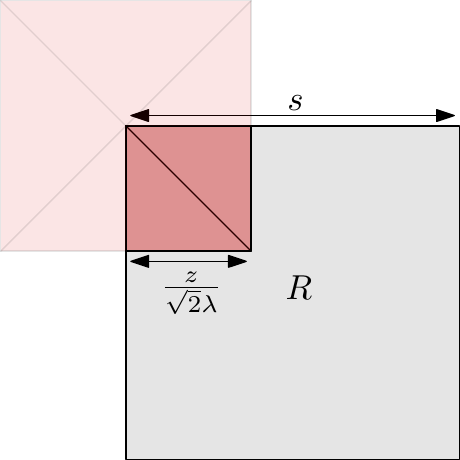}
    \caption{Top-down view of the pyramid (red) and $R$. Note that if $z < \sqrt{2}\lambda s$, then a quarter of the pyramid fits fully in $R$.}
    \label{fig:pyramid-in-cell}
\end{figure}

Now consider bound (c) in the lemma statement. Similar as above, we can argue that there is an inverse pyramid (with the same slope) going upwards from $(x, y, z)$ such that everything inside this pyramid certainly does not belong to the volume under $f$. This inverse pyramid completely contains $R$ starting at height $z + \sqrt{2}\lambda s$. Therefore, the volume $V$ under $f$ restricted to $R$ is at most $s^2 (z + \sqrt{2}\lambda s) - \frac{1}{3} s^2 \sqrt{2}\lambda s = s^2 (z + \frac{2\sqrt{2}}{3} \lambda s)$, as stated.
\end{proof}

We are now ready to give an error bound on how well the function $f_T$ obtained from the volume-based quadtree $T$ approximates the original function $f$.

\begin{lemma}\label{lem:quadtree-error}
    Let $T$ be the volume-based quadtree of $f\colon \dom \rightarrow \Reals^+$ with threshold $\rho$. Then, for any cell $v \in T$, we have that $|f(x, y) - f_T(x, y)| \leq \min\left(\frac{2\sqrt{2}}{3} \lambda s(v), \sqrt[3]{6 \lambda^2 \rho}\right)$ for all $(x, y) \in R(v)$, where $\lambda$ is the Lipschitz constant of $f$.
\end{lemma}
\begin{proof}
We first show that $f_T(x, y) - f(x, y) \leq \frac{2\sqrt{2}}{3} \lambda s(v)$. Pick any $(x, y) \in R(v)$ and let $z = f(x, y)$. From Lemma~\ref{lem:vol-bounds}(c) it follows that $V(v) \leq s(v)^2 (z + \frac{2\sqrt{2}}{3} \lambda s(v))$. But that directly implies that $h(v) \leq z + \frac{2\sqrt{2}}{3} \lambda s(v)$, and hence $f_T(x, y) - f(x, y) \leq \frac{2\sqrt{2}}{3} \lambda s(v)$.

We now show that $f(x, y) - f_T(x, y) \leq \frac{2\sqrt{2}}{3} \lambda s(v)$. We choose coordinates $(x, y) \in R(v)$ such that $f(x, y)$ is maximized in $R(v)$, and let $z = f(x, y)$. If $z \geq \sqrt{2} \lambda s(v)$, then Lemma~\ref{lem:vol-bounds}(b) states that $V(v) \geq s(v)^2 (z - \frac{2\sqrt{2}}{3} \lambda s(v))$. But then $h(v) \geq z - \frac{2\sqrt{2}}{3} \lambda s(v)$ and hence $f(x, y) - f_T(x, y) \leq \frac{2\sqrt{2}}{3} \lambda s(v)$. Otherwise, let $z = c \sqrt{2} \lambda s(v)$ for some constant $c \in [0, 1)$. From Lemma~\ref{lem:vol-bounds}(a) it follows that $V(v) \geq \frac{z^3}{6 \lambda^2} = \frac{\sqrt{2}}{3} c^3 \lambda s(v)^3$. This also implies that $h(v) \geq \frac{\sqrt{2}}{3} c^3 \lambda s(v)$. Finally note that $c - \frac{c^3}{3} \leq \frac{2}{3}$ for all $c \in [0, 1)$, and hence $f(x, y) - f_T(x, y) \leq z - h(v) \leq \sqrt{2} \lambda s(v) (c - \frac{c^3}{3}) \leq \frac{2\sqrt{2}}{3} \lambda s(v)$.

In the remainder we can assume that $\sqrt[3]{6 \lambda^2 \rho} < \frac{2\sqrt{2}}{3} \lambda s(v) < \lambda s(v)$, for otherwise the stated bound already holds. We can rewrite this inequality (by cubing and eliminating some factors) as $6 \rho < \lambda s(v)^3$ or $\frac{\rho}{s(v)^2} < \frac{1}{6} \lambda s(v)$. We then get that $h(v) = \frac{V(v)}{s(v)^2} \leq \frac{\rho}{s(v)^2} < \frac{1}{6} \lambda s(v)$. Together with the bounds already proven above, this implies that $f(x, y) < (\frac{1}{6} + \frac{2\sqrt{2}}{3}) \lambda s(v) < \sqrt{2} \lambda s(v)$ for all $(x, y) \in R(v)$. Now let $(x, y) \in R(v)$ be the coordinates that maximize $f(x, y)$ in $R(v)$, and let $z = f(x, y)$. From Lemma~\ref{lem:vol-bounds}(a) it follows that $\rho \geq V(v) \geq \frac{z^3}{6 \lambda^2}$. But then $z^3 \leq 6 \lambda^2 \rho$, or $z \leq \sqrt[3]{6 \lambda^2 \rho}$. We thus obtain that $0 \leq f(x, y) \leq \sqrt[3]{6 \lambda^2 \rho}$ for all $(x, y) \in R(v)$, and since $h(v)$ must be the average of $f(x, y)$ over all $(x, y) \in R(v)$, we directly obtain that $|f(x, y) - f_T(x, y)| \leq \sqrt[3]{6 \lambda^2 \rho}$.     
\end{proof}

Finally, we prove properties on the (spatial) neighborhood $\mathcal{N}(v)$ of a cell $v$ in a volume-based quadtree $T$. Consider the corresponding function $f_T$. Note that a (weak) local maximum of $f_T$ corresponds to a cell $v \in T$ such that $h(v) \geq h(w)$ for all $w \in \mathcal{N}(v)$. To verify that property efficiently, we would like to show that $|\mathcal{N}(v)|$ is bounded for any local maximum $v \in T$. This is nearly true, as we show below.

\begin{lemma}\label{lem:cell-local-maxima}
Let $T$ be the volume-based quadtree of $f\colon \dom \rightarrow \Reals^+$ with threshold $\rho$, and let $\lambda$ be the Lipschitz constant of $f$. If a cell $v \in T$ satisfies $\frac{\lambda s(v)^3}{\rho} \leq \frac{90}{\sqrt{2}}$, then for all $w \in \mathcal{N}(v)$ with $h(v) \geq h(w)$ it holds that $s(w) \geq \frac{1}{4} s(v)$.
\end{lemma}
\begin{proof}
Let $v \in T$ be a leaf node in $T$ with $\frac{\lambda s(v)^3}{\rho} \leq \frac{90}{\sqrt{2}}$ and let $w \in \mathcal{N}(v)$ be a neighboring cell of $v$ with $h(v) \geq h(w)$. We get that $h(w) \leq h(v) \leq \frac{\rho}{s(v)^2}$. In particular, there must be coordinates $(x, y) \in R(w)$ such that $f(x, y) \leq \frac{\rho}{s(v)^2}$. Now consider the parent $u$ of $w$ in $T$. Since $R(w) \subset R(u)$, we have $(x, y) \in R(u)$. Thus, we can apply Lemma~\ref{lem:vol-bounds}(c) to obtain that $V(u) \leq s(u)^2 (\frac{\rho}{s(v)^2} + \frac{2 \sqrt{2}}{3} \lambda s(u))$. We now write $s(u) = \beta s(v)$ for some constant $\beta > 0$. As $V(u) > \rho$, we obtain the following inequality:
    \begin{alignat*}{2}
        && \beta^2 s(v)^2 \left(\frac{\rho}{s(v)^2} + \frac{2 \sqrt{2}}{3} \lambda \beta s(v)\right)
        &> \rho\\
        &\Rightarrow\qquad 
        & \beta^2 \rho + \frac{2 \sqrt{2}}{3} \lambda \beta^3 s(v)^3
        &> \rho \\
        &\Rightarrow\qquad
        & \frac{2 \sqrt{2}}{3} \lambda \beta^3 s(v)^3
        &> (1 - \beta^2)\rho \\
        &\Rightarrow\qquad
        & \frac{\lambda s(v)^3}{\rho} 
        &> \frac{3}{2\sqrt{2}}\frac{1 - \beta^2}{\beta^3} \\
    \end{alignat*}
Since $\frac{\lambda s(v)^3}{\rho} \leq \frac{90}{\sqrt{2}}$, we obtain that $\frac{3}{2\sqrt{2}}\frac{1 - \beta^2}{\beta^3} < \frac{90}{\sqrt{2}}$, or $\frac{1 - \beta^2}{\beta^3} < 60$. It is easy to verify that this inequality only holds for $\beta > \frac{1}{4}$. We can then directly conclude that $s(w) = \frac{1}{2}s(u) = \frac{1}{2} \beta s(v) > \frac{1}{8} s(v)$. As the ratios of sizes between cells in $T$ must always be a power of $2$, we conclude that $s(w) \geq \frac{1}{4} s(v)$.
\end{proof}

The result of Lemma~\ref{lem:cell-local-maxima} can be interpreted as follows. If a cell $v \in T$ is not very large, then it can only be a local maximum if for all cells $w \in \mathcal{N}(v)$ it holds that $s(w) \geq \frac{1}{4} s(v)$. In that case, we get that $|\mathcal{N}(v)| \leq 20$. This makes it possible to efficiently check if a cell $v \in T$ is a local maximum, assuming that $s(v)$ is not too large.

\section{From volume to points}\label{sec:disc}

We aim to maintain a volume-based quadtree for $f = \KDE_P$ over time. Most common kernels, with the exception of the uniform kernel, are Lipschitz continuous and hence the resulting KDE is also Lipschitz continuous. We scale $\KDE_P$ such that the volume underneath $\KDE_P$ is $1$. Additionally, we assume the kernel width $\sigma$ to be $1$. As a result, for common kernels, such as the cone kernel, $f = \KDE_P$ is Lipschitz continuous with a small Lipschitz constant. Furthermore, the maximum value of $f$ is bounded as well.

We want to approximate the volume under $f = \KDE_P$, as the points in $P$ are moving, via a (small) set of moving points $Q$. We use $V_R(f(t))$ to refer to the volume under a function $f$ at time $t$ restricted to a region $R$. For a chosen value $\ec > 0$, we require the following property on $Q$: for any square region $R \subseteq \dom$ and time $t$, we have that $\left|\frac{|Q \cap R|}{|Q|} - V_R(\KDE_P(t))\right| < \ec$. We plan to use $\varepsilon$-approximations to construct a suitable point set $Q$. An $\varepsilon$-approximation needs an initial point set from which to construct $Q$. We therefore first take a dense point sample $S$ under each kernel $K$ to represent its volume (see Section~\ref{sec:kernel-approx}). Then we combine the samples for the individual kernels into a set $\mathcal{S}$ which serves as the input for the $\varepsilon$-approximation that will ultimately result in $Q$ (see Section~\ref{sec:coreset}). In Section~\ref{sec:weight-quad} we then show how to replace the actual volume under $\KDE_P$ with the points in $Q$ when constructing the volume-based quadtree.


\subsection{Approximating a single kernel}\label{sec:kernel-approx}
Let $K\colon[-1, 1]^2 \rightarrow \Reals^+$ denote the kernel function. We aim to represent the volume under $K$ using a set of points $S$. For ease of exposition we assume in the remainder of the paper that the maximum value of the kernel $K$ is bounded by 1 (this holds for most common kernels). Note that we can ignore the time component for this approximation, as a single kernel represents only a single point with a single trajectory over time (if the approximation bound holds for all square regions at a single time $t$, then it also holds for other times by simply shifting the squares). 

To obtain $S$, we consider a regular $r \times r$ grid $G$ on the domain $[-1, 1]^2$, for some value $r$ to be chosen later. Note that the area of a single grid cell $c \in G$ is $\frac{4}{r^2}$. We construct a \emph{grid-based sampling} $S(r)$ of $K$ by arbitrarily placing $\lceil r z(c)\rceil$ points in every cell $c \in G$, where $z(c)$ is the average value of $K$ in the corresponding grid cell $c$. See Fig.~~\ref{fig:kernel-approx} for an example. We can prove the following property on $S(r)$.

\begin{figure}
    \centering
    \includegraphics{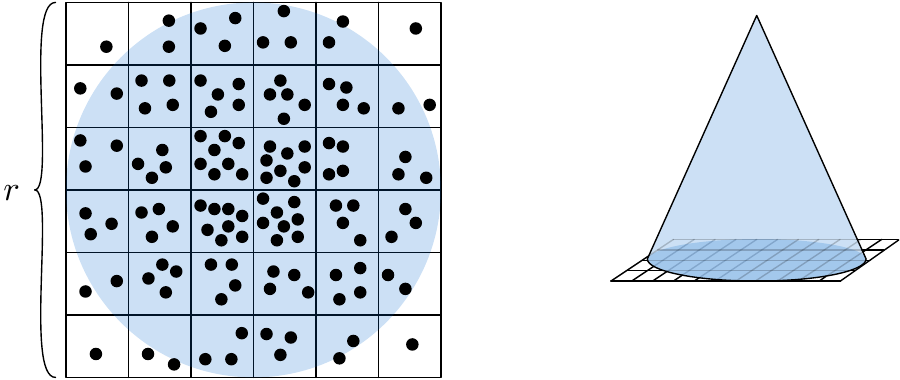}
    \caption{An example of how a cone kernel can be approximated by a point set. In each grid cell, $\lceil r z(c) \rceil$ points are arbitrarily placed.}
    \label{fig:kernel-approx}
\end{figure}

\begin{restatable}{lemma}{discrete}\label{lem:discrete}
Let $K\colon[-1, 1]^2 \rightarrow \Reals^+$ be a function such that the total volume under $K$ is $1$ and $K(x, y) \leq 1$ for all $(x, y) \in [-1, 1]^2$. If $S(r)$ is a grid-based sampling of $K$ with parameter $r$, then for any square region $R$ that overlaps with the domain of $K$ we have that $\left|\frac{|S(r) \cap R|}{|S(r)|} - V_R(K)\right| \leq \frac{36}{r}$.
\end{restatable}
\begin{proof}
We first show that $\frac{r^3}{4} \leq |S(r)| \leq \frac{r^3}{4} + r^2$. The volume under $K$ restricted to a single grid cell $c \in G$ is $\frac{4 z(c)}{r^2}$, and hence $\sum_{c \in G} \frac{4 z(c)}{r^2} = 1$ or $\sum_{c \in G} z(c) = \frac{r^2}{4}$. Since we place $\lceil r z(c)\rceil$ points per grid cell, we have that $|S(r)| = \sum_{c \in G} \lceil r z(c)\rceil$. Therefore, $\sum_{c \in G} r z(c) \leq |S(r)| \leq \sum_{c \in G} (r z(c) + 1)$. This directly implies a lower bound of $\frac{r^3}{4}$ on $|S(r)|$. For the upper bound we get that $\sum_{c \in G} (r z(c) + 1) = r \sum_{c \in G} z(c) + |G| = \frac{r^3}{4} + r^2$, as claimed. 

Now let $V(c) = \frac{4 z(c)}{r^2}$ indicate the volume under $K$ restricted to a grid cell $c \in G$. For any grid cell $c \in G$ we get that $\frac{\lceil r z(c)\rceil}{|S(r)|} \geq \frac{r z(c)}{\frac{r^3}{4} + r^2} = \frac{1}{1 + \frac{4}{r}} \frac{4 z(c)}{r^2} = \frac{1}{1 + \frac{4}{r}} V(c)$. On the other hand we have that $\frac{\lceil r z(c)\rceil}{|S(r)|} \leq \frac{r z(c) + 1}{\frac{1}{4} r^3} = \frac{4 z(c)}{r^2} + \frac{4}{r^3} = V(c) + \frac{4}{r^3}$.

Now let $R$ be any square region that overlaps $[-1, 1]^2$. If a grid cell $c \in G$ lies completely outside of $R$, then the corresponding points in $c$ are excluded from $S(r) \cap R$, and the volume under $K$ restricted to $c$ is not part of $V_R(K)$, and hence no error is made with respect to this grid cell. Now consider the set of grid cells $C \subseteq G$ that lie completely within $R$. The error with respect to those grid cells is $|\sum_{c \in C}(\frac{\lceil r z(c)\rceil}{|S(r)|} - V(c))|$. Using the bounds above, we get that $\frac{\lceil r z(c)\rceil}{|S(r)|} \geq \frac{1}{1 + \frac{4}{r}} V(c)$ and thus $\sum_{c \in C}(V(c) - \frac{\lceil r z(c)\rceil}{|S(r)|}) \leq V(C) (1 - \frac{1}{1 + \frac{4}{r}}) = \frac{4}{r+4} V(C)$, where $V(C)$ is the total volume under $K$ for all grid cells $c \in C$. Since $V(C) \leq 1$, this gives an additive error of at most $\frac{4}{r+4} \leq \frac{4}{r}$. On the other hand, we have that $\frac{\lceil r z(c)\rceil}{|S(r)|} \leq V(c) + \frac{4}{r^3}$ and thus $\sum_{c \in C}(\frac{\lceil r z(c)\rceil}{|S(r)|} - V(c)) \leq |C| \frac{4}{r^3}$. Since $|C| \leq r^2$ we again obtain an additive error of at most $\frac{4}{r}$. Finally, consider the grid cells that only partially overlap with $R$. The error for such a cell $c \in G$ is bounded by $\max(V(c), \frac{\lceil r z(c) \rceil}{|S(r)|}) \leq V(c) + \frac{4}{r^3}$. It is easy to see that there can be at most $4 r$ grid cells that partially overlap with $R$. Since by construction, $z(c) \leq 1$, we must have $V(c) \leq \frac{4}{r^2}$. This means that the total error with respect to these cells is at most $4 r (\frac{4}{r^2} + \frac{4}{r^3}) \leq \frac{32}{r}$ (for $r \geq 1$). Thus, the total error is at most $\frac{4}{r} + \frac{32}{r} = \frac{36}{r}$ as claimed. 
\end{proof}

Now, for a chosen error $\es > 0$, we can simply choose $r = \frac{36}{\es}$ to obtain a grid-based sampling $S$ with $O(\frac{1}{\es^3})$ points that approximates the volume under $K$ with error at most $\es$, according to Lemma~\ref{lem:discrete}.

\subsection{Coreset}\label{sec:coreset}
We now use the results of Section~\ref{sec:kernel-approx} to construct an approximation for the volume under $\KDE_P$, as the points in $P$ are moving. For a chosen error $\es > 0$, we construct a grid-based sampling $S_p$ of $O(\frac{1}{\es^3})$ points around each point $p \in P$, resulting in $O(\frac{n}{\es^3})$ points in total. We let the points in $S_p$ move in the same direction as the corresponding point $p \in P$. Note that the complete set of points $\mathcal{S} = \bigcup_{p \in P} S_p$ provides an approximation for the volume under $\KDE_P$ with error at most $\es$ for all times $t$.

We now use the algorithm by Agarwal~\etal~\cite{approx-rangespace} to construct an $\varepsilon$-approximation $Q$ of $\mathcal{S}$. For completeness, we briefly review the algorithm here in order to apply it to our setting. To compute an $\varepsilon$-approximation for a range space $X = (S, \mathcal{R})$, they first build a balanced binary tree on the points in $S$. Then, the $\varepsilon$-approximation is computed in a bottom-up fashion, where at each node in the tree an $\varepsilon$-approximation is computed of the points in the subtree rooted at that node. To compute the $\varepsilon$-approximation at a node in the tree, the $\varepsilon$-approximations of the two child nodes are first simply merged. This does not introduce an error. Then, if the newly obtained $\varepsilon$-approximation contains more than $\mu = \frac{c}{\varepsilon^2} (\log n \log\left(\frac{\log n}{\varepsilon}\right))^2$ points (for some large enough constant $c$), a halving step is performed which reduces the size of the $\varepsilon$-approximation by half. This can be done in $O(\mu^{d+1})$ time, where $d$ is the VC-dimension of $X$, and doing so introduces an error of $O(\frac{\log \mu}{\sqrt{\mu}})$. Finally, the root contains an $\varepsilon$-approximation for $X$, but may contain too many points. Further halving steps are then applied until the $\varepsilon$-approximation has size $O(\frac{1}{\varepsilon^2} \log\left(\frac{1}{\varepsilon}\right))$. The result is summarized in the lemma below.

\begin{lemma}[\hspace{-1sp}\cite{approx-rangespace}]\label{lem:agarwal-coreset}
    Given a range space $X = (S, \mathcal{R})$ of VC-dimension $d$ and a parameter $\varepsilon$, we can compute an $\varepsilon$-approximation of $X$ of size $O(\frac{1}{\varepsilon^2}\log(\frac{1}{\varepsilon}))$ in time $O(n\mu^d)$, where $|S| = n$ and $\mu = \frac{c}{\varepsilon^2}\left(\log n \log\left(\frac{\log n}{\varepsilon}\right)\right)^2.$
\end{lemma}

Now assume we aim to compute a coreset $Q$ that approximates the volume under $\KDE_P$ over time with an additive error of $\ec > 0$. We first compute a grid-based sampling $S_p$ for a single point $p \in P$ with $\es = \frac{\ec}{3}$. Let $\mathcal{S} = \bigcup_{p \in P} S_p$ be the total set of sampling points. Next, we replace a single grid-based sampling $S_p$ by an $\varepsilon$-approximation $S'_p$ of $S_p$ by running the algorithm of Agarwal~\etal on $\frac{1}{\es^3}$ points with $\varepsilon = \frac{\ec}{3}$. We use a copy of the resulting $\varepsilon$-approximation $S'_p$ for all other points in $P$ as well, resulting in a total set $\mathcal{S}'$ consisting of $O(\frac{n}{\varepsilon^2} \log\left(\frac{1}{\varepsilon}\right))$ points. Then, we again compute an $\varepsilon$-approximation $Q$ with the algorithm of Agarwal~\etal, but now on $\mathcal{S}'$, again with $\varepsilon = \frac{\ec}{3}$. As we already have $\varepsilon$-approximations for the grid-based samplings of individual points, we perform only at most $n$ halving steps.

Thus, assuming $n > \frac{1}{\es^3}$, by Lemma~\ref{lem:agarwal-coreset} we can compute $Q$ in $O(n \mu^d)$ time, where $d$ is the VC-dimension of $X = (\mathcal{S}, \mathcal{R})$. The resulting set $Q$ has size $O(\frac{1}{\ec^2} \log\left(\frac{1}{\ec}\right))$ and has an additive error of $\frac{\ec}{3}$ with respect to $\mathcal{S'}$, which has an additive error of $\frac{\ec}{3}$ with respect to $\mathcal{S}$. Since $\mathcal{S}$ approximates the volume under $\KDE_P$ with error at most $\frac{\ec}{3}$, we obtain that $Q$ approximates the volume under $\KDE_P$ with additive error at most $\ec$. 

To ensure that this algorithm works, we need to show that the range space $X = (\mathcal{S}, \mathcal{R})$ has bounded VC-dimension, where $\mathcal{R}$ contains all subsets of points in $\mathcal{S}$ that may appear in a square region $R$ at some time $t$. Note that this is non-trivial, since the points in $\mathcal{S}$ correspond to moving points (see Fig.~\ref{fig:3d-rangespace}). As already stated earlier, we assume that the points follow linear motion. We first establish a bound on the VC-dimension for points moving in $1$ dimension, before extending the result to points moving in $2$ dimensions.

\begin{lemma}\label{lem:VC-1dim}
Let $X_1 = (S_1, \mathcal{R}_1)$ be a range space where $S_1$ contains a set of $x$-monotone lines in $\Reals^2$, and $\mathcal{R}_1$ contains all subsets of lines in $S_1$ that can be intersected by a vertical line segment in $\Reals^2$. The VC-dimension of $X_1$ is $5$.     
\end{lemma}
\begin{proof}
Consider the geometric point-line dual of $S_1$, where $a x - b$ is mapped to the point $(a, b)$ and vice versa. In that representation $S_1$ corresponds to a set of points, and a vertical segment in the primal corresponds to an infinite strip bounded by two parallel (non-vertical) lines in the dual. We can thus consider the range space $(S_1', \mathcal{R}_1')$ where $S_1'$ consists of a set of points in $\Reals^2$, and $\mathcal{R}_1'$ consists of the subsets of points that exactly lie in an infinite (non-vertical) strip.

We show that the VC-dimension of $(S_1', \mathcal{R}_1')$ is equal to $5$, which directly implies that the VC-dimension of $(S_1, \mathcal{R}_1)$ is also equal to $5$. To this end, we show that infinite strips can shatter a set of $5$ points, but not a set of $6$ points. It is easy to verify that a set of $5$ points placed at the corners of a regular pentagon can be shattered by infinite strips. We thus focus on the fact that a set of $6$ points can never be shattered by infinite strips.

First assume that the points are not in convex position. In that case, there must exist a point $p$ and three other points $p_1$, $p_2$, and $p_3$, such that $p$ lies in the convex hull of $p_1$, $p_2$, and $p_3$. Now consider an infinite strip that contains $p_1$, $p_2$, and $p_3$. Since the infinite strip is convex, it must also contain $p$. Hence, the set $\{p_1, p_2, p_3\}$ is not in the range space and the set of points is not shattered (see Fig.~\ref{subfig:VC-dim-nocon} for an example).

\begin{figure}
    \begin{subfigure}[t]{.49\textwidth}
        \centering
        \includegraphics[page=1]{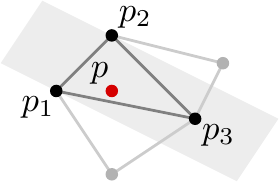}
        \subcaption{Not in convex position}
        \label{subfig:VC-dim-nocon}
    \end{subfigure}
    \begin{subfigure}[t]{.49\textwidth}
        \centering
        \includegraphics[page=2]{figures/VC-dimension-proof.pdf}
        \subcaption{Convex position}
        \label{subfig:VC-dim-con}
    \end{subfigure}
\caption{Examples of two sets of points in which the red point prevents shattering.}
\label{fig:VC-dim-proof}
\end{figure}

Now assume that the points are in convex position and let $p_1, \ldots, p_6$ be the points in clockwise order along the convex hull. We show that the set $\{p_1, p_3, p_5\}$ cannot be the only points contained in an infinite strip (see~Fig.~\ref{subfig:VC-dim-con}). Assume for the sake of contradiction that there exists an infinite strip containing only the points $\{p_1, p_3, p_5\}$. This implies that at least two points of the set $\{p_2, p_4, p_6\}$ must lie on the same side outside of the strip; assume without loss of generality that this holds for $p_2$ and $p_4$. Since the points $\{p_2, p_4, p_6\}$ lie strictly outside of the infinite strip, we can always widen the strip slightly so that $p_1$, $p_3$, and $p_5$ lie strictly in the interior of the strip. But then the bounding line of the strip separating $\{p_1, p_3, p_5\}$ from $\{p_2, p_4\}$ intersects the boundary of the convex hull more than twice (it must intersect all edges in the chain $p_1$--$p_2$--$p_3$--$p_4$--$p_5$), which is a contradiction. Hence, no set of $6$ points can be shattered by infinite strips. Thus we can conclude that the VC-dimension of $(S_1, \mathcal{R}_1)$ is $5$.
\end{proof}

To extend the result of Lemma~\ref{lem:VC-1dim} to $2$ dimensions, we need Sauer's lemma~\cite{sauer1972density}. Let the \emph{growth function} be defined as:
\begin{equation}
\mathcal{G}_d(n) = \sum_{i=0}^d {n \choose i}.
\end{equation}

\begin{lemma}[Sauer's lemma]\label{lem:Sauer}
If $(S, \mathcal{R})$ is a range space of VC-dimension $d$ with $|S| = n$, then $|\mathcal{R}| \leq \mathcal{G}_d(n)$. 
\end{lemma}

\begin{figure}[b]
    \centering
    \includegraphics{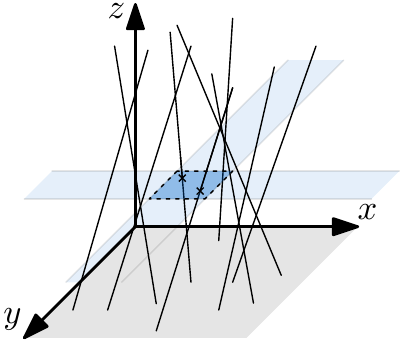}
    \caption{An example of the range space $X = (\mathcal{S}, \mathcal{R})$ plotted over time. The lines denote set $\mathcal{S}$, and the blue square is an example of a range in $\mathcal{R}$.}
    \label{fig:3d-rangespace}
\end{figure}

To find the VC-dimension of range space $X = (\mathcal{S}, \mathcal{R})$, we now consider the linearly moving points in $\mathcal{S}$ plotted over time on the $z$-axis. This gives us a set of $z$-monotone lines in $\Reals^3$ representing $\mathcal{S}$. The ranges $\mathcal{R}$ are now described by $x,y$-aligned squares projected on some $z$-plane. See Fig.~\ref{fig:3d-rangespace} for an example.

\begin{restatable}{lemma}{VCdim}\label{lem:VC-2dim}
Let $X_2 = (S_2, \mathcal{R}_2)$ be a range space where $S_2$ contains a set of $z$-monotone lines in $\Reals^3$, and $\mathcal{R}_2$ contains all subsets of lines in $S_2$ that can be intersected by an axis-aligned square in $\Reals^3$ with constant $z$-coordinate. The VC-dimension of $X_2$ is at most $38$.     
\end{restatable}

\begin{proof}
Let $X'_2 = (S_2, \mathcal{R}'_2)$ be an alternative range space where $\mathcal{R}'_2$ contains all subsets of lines in $S_2$ that can be intersected by an axis-aligned strip in $\Reals^3$ which has constant $z$-coordinate and extends infinitely along the $x$-axis. In that case, the $x$-coordinates of the lines in $S_2$ are irrelevant, and we can observe that $X'_2$ actually corresponds to a $1$-dimensional range space $X_1$. Thus, $X'_2$ has VC-dimension $5$ by Lemma~\ref{lem:VC-1dim}. Similarly we can create a range space $X''_2 = (S_2, \mathcal{R}''_2)$, which is similar to $X'_2$, but then with strips extending infinitely along the $y$-axis instead of the $x$-axis. Again, we can conclude that $X''_2$ has VC-dimension $5$ by Lemma~\ref{lem:VC-1dim}. Now note that every range in $\mathcal{R}_2$ can be obtained by taking the intersection of a range in $\mathcal{R}'_2$ and a range in $\mathcal{R}''_2$ (every square is the intersection of two infinite strips), that is, $\mathcal{R}_2 \subseteq \{R' \cap R''\mid R' \in \mathcal{R}'_2, R'' \in \mathcal{R}''_2\}$. Now assume that $X'2$ has VC-dimension $d$. Then there must be a subset $Y \subseteq S_2$ with $|Y|=d$ that is shattered by $\mathcal{R}_2$, and thus $|\mathcal{R}_2| \geq 2^d$, also when the ground set $S_2$ is restricted to $Y$. We also have that $|\mathcal{R}_2| \leq |\mathcal{R}'_2| |\mathcal{R}''_2| \leq (\mathcal{G}_5(d))^2$ by Lemma~\ref{lem:Sauer}, if we restrict the ground set $S_2$ to $Y$. We thus obtain that $2^d \leq (\mathcal{G}_5(d))^2$. It is easy to verify that this inequality holds for $d = 38$, but not for $d = 39$.
\end{proof}


From Lemma~\ref{lem:VC-2dim} we can directly conclude that the range space $X = (\mathcal{S}, \mathcal{R})$ has VC-dimension $38$. In the remainder of this paper we will simply refer to the set of linearly moving points $Q$ as the coreset of $\KDE_P$, where the additive error with respect to the volume is $\ec$. We summarize the result in the following lemma.

\begin{restatable}{lemma}{coreset}\label{lem:coreset}
Let $\KDE_P$ be a KDE function on a set of $n$ linearly moving points $P$. For any $\ec > 0$, we can construct a coreset $Q$ of linearly moving points such that, for any time $t$ and any square region $R$, we get that $\left|\frac{|Q \cap R|}{|Q|} - V_R(\KDE_P(t))\right| < \ec$, where $V_R(\KDE_P(t))$ is the volume under $\KDE_P$ at time $t$ restricted to $R$. $Q$ consists of $O(\frac{1}{\ec^2} \log\left(\frac{1}{\ec}\right))$ points and can be constructed in $O(n \poly\left(\frac{\log n}{\ec}\right))$ time.
\end{restatable}

\subsection{Weight-based quadtree}\label{sec:weight-quad}
The coreset $Q$ of $\KDE_P$ functions as a proxy for the volume under $\KDE_P$ restricted to some square region. We can therefore approximate the volume-based quadtree $T$ of $\KDE_P$ with the \emph{weight-based quadtree} $\WT$ of $Q$, which is defined as follows for any fixed time $t$ and volume threshold $\rho$. The root $r \in \WT$ again corresponds to the whole domain $\dom = [0, D]^2$ of $\KDE_P$. Then, starting from the root, we subdivide a node $v \in \WT$ when the fraction $\frac{|Q \cap R(v)|}{|Q|}$ exceeds $\rho$, and we recursively apply this rule to all newly created child nodes. However, we do not subdivide nodes with $s(v) \leq \sqrt{\rho}$, so that the lower bound on cell size in Lemma~\ref{lem:cellsize} is preserved in $\WT$. We refer to $\frac{|Q \cap R(v)|}{|Q|}$ as the \emph{weight} of cell $v \in \WT$, denoted by $W(v)$. As a result, for every cell (leaf) $v \in \WT$, the weight of $v$ is at most $\rho + \ec$ when $s(v) \leq \sqrt{\rho}$, and at most $\rho$ otherwise. Finally, for every cell $v \in \WT$, we set the value $h(v)$ to $\frac{W(v)}{s(v)^2}$.

By construction, the minimum cell size and maximum depth in Lemma~\ref{lem:cellsize} and Corollary~\ref{cor:quaddepth} are preserved by $\WT$. Since the total weight of all cells in $\WT$ is $1$ by construction, the number of nodes in Lemma~\ref{lem:treecomplexity} also holds for $\WT$. However, the error made by the volume-based quadtree in Lemma~\ref{lem:quadtree-error} does not directly hold for $\WT$, as we need to incorporate the error on the volume under the function $f = \KDE_P$. We therefore give a new bound on the error for $\WT$.

\begin{restatable}{lemma}{weightquaderror}\label{lem:weight-quad-error}
Let $f = \KDE_P$ be a KDE function on a set of linearly moving points $P$ at a fixed time $t$, and let $Q$ be a coreset for $\KDE_P$ with additive error $\ec$. Furthermore, let $\WT$ be the weight-based quadtree on $Q$ with threshold $\rho$ at the same time $t$. Then, for any cell $v \in \WT$, we have that $|f(x, y) - f_{\WT}(x, y)| < \min\left(\frac{2 \sqrt{2}}{3} \lambda s(v), \sqrt[3]{6 \lambda^2 (\rho + 2 \ec)}\right) + \frac{\ec}{s(v)^2}$ for all $(x, y) \in R(v)$, where $\lambda$ is the Lipschitz constant of $f$.     
\end{restatable}

\begin{proof}
In this proof we use $V(v)$ to refer to the actual volume under $f$ restricted to $R(v)$.

We first show that $f_{\WT}(x, y) - f(x, y) \leq \frac{2\sqrt{2}}{3} \lambda s(v) + \frac{\ec}{s(v)^2}$. Pick any $(x, y) \in R(v)$ and let $z = f(x, y)$. From Lemma~\ref{lem:vol-bounds} (c) it follows that $V(v) \leq s(v)^2 (z + \frac{2\sqrt{2}}{3} \lambda s(v))$. By construction of $\WT$ we also obtain that $W(v) \leq s(v)^2 (z + \frac{2\sqrt{2}}{3} \lambda s(v)) + \ec$. That directly implies that $h(v) \leq z + \frac{2\sqrt{2}}{3} \lambda s(v) + \frac{\ec}{s(v)^2}$, and hence $f_{\WT}(x, y) - f(x, y) \leq \frac{2\sqrt{2}}{3} \lambda s(v) + \frac{\ec}{s(v)^2}$.

We now show that $f(x, y) - f_{\WT}(x, y) \leq \frac{2\sqrt{2}}{3} \lambda s(v) + \frac{\ec}{s(v)^2}$. We choose coordinates $(x, y) \in R(v)$ such that $f(x, y)$ is maximized and let $z = f(x, y)$. If $z \geq \sqrt{2} \lambda s(v)$, then Lemma~\ref{lem:vol-bounds} (b) states that $V(v) \geq s(v)^2 (z - \frac{2\sqrt{2}}{3} \lambda s(v))$. By construction of $\WT$ we also obtain that $W(v) \geq s(v)^2 (z - \frac{2\sqrt{2}}{3} \lambda s(v)) - \ec$. But then $h(v) \geq z - \frac{2\sqrt{2}}{3} \lambda s(v) - \frac{\ec}{s(v)^2}$ and hence $f(x, y) - f_{\WT}(x, y) \leq \frac{2\sqrt{2}}{3} \lambda s(v) + \frac{\ec}{s(v)^2}$. Otherwise, let $z = c \sqrt{2} \lambda s(v)$ for some constant $c \in [0, 1)$. From Lemma~\ref{lem:vol-bounds} (a) it follows that $V(v) \geq \frac{z^3}{6 \lambda^2} = \frac{\sqrt{2}}{3} c^3 \lambda s(v)^3$. By construction of $\WT$ we also obtain that $W(v) \geq \frac{\sqrt{2}}{3} c^3 \lambda s(v)^3 - \ec$. 
This directly implies that $h(v) \geq \frac{\sqrt{2}}{3} c^3 \lambda s(v) - \frac{\ec}{s(v)^2}$. Finally note that $c - \frac{c^3}{3} \leq \frac{2}{3}$ for all $c \in [0, 1)$, and hence $f(x, y) - f_{\WT}(x, y) \leq z - h(v) \leq \sqrt{2} \lambda s(v) (c - \frac{c^3}{3}) + \frac{\ec}{s(v)^2} \leq \frac{2\sqrt{2}}{3} \lambda s(v) + \frac{\ec}{s(v)^2}$.

In the remainder we can assume that $\sqrt[3]{6 \lambda^2 (\rho + 2 \ec)} < \frac{2\sqrt{2}}{3} \lambda s(v) < \lambda s(v)$, for otherwise the stated bound already holds. We can rewrite this inequality (by cubing and eliminating some factors) as $6 (\rho + 2 \ec) < \lambda s(v)^3$ or $\frac{\rho + \ec}{s(v)^2} < \frac{1}{6} \lambda s(v) - \frac{\ec}{s(v)^2}$. We then get that $h(v) = \frac{W(v)}{s(v)^2} \leq \frac{\rho + \ec}{s(v)^2} < \frac{1}{6} \lambda s(v) - \frac{\ec}{s(v)^2}$. Together with the bounds already proven above, this implies that $f(x, y) < (\frac{1}{6} + \frac{2\sqrt{2}}{3}) \lambda s(v) + \frac{\ec}{s(v)^2} - \frac{\ec}{s(v)^2} < \sqrt{2} \lambda s(v)$ for all $(x, y) \in R(v)$. Now let $(x, y) \in R(v)$ be the coordinates that maximize $f(x, y)$ and let $z = f(x, y)$. From Lemma~\ref{lem:vol-bounds} (a) it follows that $V(v) \geq \frac{z^3}{6 \lambda^2}$. By construction of $\WT$ we have that $V(v) \leq W(v) + \ec \leq \rho + 2\ec$. But then $z^3 \leq 6 \lambda^2 (\rho + 2 \ec)$ or $z \leq \sqrt[3]{6 \lambda^2 (\rho + 2 \ec)}$. We thus obtain that $0 \leq f(x, y) \leq \sqrt[3]{6 \lambda^2 (\rho + 2 \ec)}$ for all $(x, y) \in R(v)$. This directly implies that $V(v) \leq s(v)^2 \sqrt[3]{6 \lambda^2 (\rho + 2 \ec)}$, and by construction of $\WT$, that $W(v) \leq s(v)^2 \sqrt[3]{6 \lambda^2 (\rho + 2 \ec)} + \ec$. Thus we get that $h(v) \leq \sqrt[3]{6 \lambda^2 (\rho + 2 \ec)} + \frac{\ec}{s(v)^2}$, which directly implies that $|f(x, y) - f_{\WT}(x, y)| \leq \sqrt[3]{6 \lambda^2 (\rho + 2 \ec)} + \frac{\ec}{s(v)^2}$.     
\end{proof}

We may now choose parameters such that, for any error $\varepsilon > 0$, we get that $|\KDE_P(x, y) - f_{\WT}(x, y)| < \varepsilon$ for every time $t$. Specifically, we choose $\rho = \frac{\varepsilon^3}{6 \lambda^2 (8 + 2 \varepsilon z^{*})} = \Theta(\varepsilon^3)$ and $\ec = \frac{\varepsilon^4 z^{*}}{48 \lambda^2 (8 + 2 \varepsilon z^{*})} = \Theta(\varepsilon^4)$, where $\lambda = \Theta(1)$ and $z^{*} \leq 1$ are the Lipschitz constant and maximum value of $\KDE_P$, respectively. Observe that $\rho + 2 \ec = \frac{8 \varepsilon^3 + 2 \varepsilon^4 z^{*}}{48 \lambda^2 (8 + 2 \varepsilon z^{*})} = \frac{\varepsilon^3}{48 \lambda^2}$, and hence $\sqrt[3]{6 \lambda^2 (\rho + 2 \ec)} = \frac{\varepsilon}{2}$. Furthermore, by Lemma~\ref{lem:cellsize} we know that $s(v)^2 \geq \frac{1}{4} \rho z^{*}$ for all $v \in \WT$. As $\frac{\ec}{\rho} = \frac{\varepsilon z^{*}}{8}$, we get that $\frac{\ec}{s(v)^2} \leq \frac{\varepsilon}{2}$ for all $v \in \WT$. With these choices of $\rho$ and $\ec$, it now follows from Lemma~\ref{lem:weight-quad-error} that $|\KDE_P(x, y) - f_{\WT}(x, y)| < \varepsilon$ for every time $t$. This also implies that the weight-based quadtree $\WT$ has at most $O(\frac{1}{\varepsilon^3} \log \left(\frac{n}{\varepsilon}\right))$ nodes (Lemma~\ref{lem:treecomplexity}), and that the coreset $Q$ contains $O(\frac{1}{\varepsilon^8} \log\left(\frac{1}{\varepsilon}\right))$ points in total (Lemma~\ref{lem:coreset}). In the remainder of this paper we assume that $Q$ and $\WT$ are constructed with the parameters $\rho$ and $\ec$ chosen above.

\section{KDS for density approximation}\label{sec:kinetic}
In this section we describe a KDS to efficiently maintain the weight-based quadtree $\WT$ on a set of linearly moving points $Q$ over time. By the results of Section~\ref{sec:disc} and Lemma~\ref{lem:persistent-maxima}, keeping track of the local maxima of $f_{\WT}$ is sufficient to track the local maxima of $\KDE_P$ with persistence at least $2 \varepsilon$. We therefore store for every cell $v \in \WT$ whether it is a local maximum or not. In addition, we store a set of pointers $\SN(v)$ in each node $v\in\WT$ (including internal nodes) to all nodes $w\in\WT$ with $\frac{1}{4}s(v) \leq s(w) \leq 4 s(v)$ such that $R(v)$ and $R(w)$ share (a piece of) boundary. It is easy to see that $|\SN(v)| = O(1)$ for all $v \in \WT$. These pointers will be used to efficiently update whether a cell is a local maximum or not.

\subsection{Event Handling}

Assume that a point $q \in Q$ moved from a cell $v \in \WT$ to a cell $u \in \mathcal{N}(v)$. To determine the cell $u$ to which the point $q$ has moved, we can simply use a point location query on $\WT$. Next, we update the weights $W(v)$ and $W(u)$ accordingly. If $W(u) > \rho$ after the update, then we must split the cell $u$ into four cells, and compute the weights of the children of $u$. Possibly we also need to split a child of $u$ recursively, but this can apply to only one child of $u$. Furthermore, let $w$ be the parent of $v$. If the sum of the weights of the children of $w$ become $\leq \rho$, then we need to remove the children of $w$ (we merge $w$) and compute a new weight $W(w)$ for $w$. Note that this is only possible if all children of $w$ are leaves in $\WT$. We may also need to merge the parent of $w$ recursively, so we check this as well. 

If a cell $u \in \WT$ is split, then we need to compute $\SN(w)$ for each new child $w$ of $u$. Note that $\SN(w)$ consists of a subset of $\SN(u)$ and children of nodes in $\SN(u) \cup \{u\}$, which consists of $O(1)$ nodes of $\WT$ in total. We can thus compute $\SN(w)$ (and add pointers to $w$ for nodes in $\SN(w)$) in $O(1)$ time for each child $w$ of $u$. If a cell $v \in \WT$ is removed (due to a merge), then we simply need to remove pointers to $v$ for all nodes in $\SN(v)$. To actually perform a split on a cell $u \in \WT$ we must reassign the points in $Q \cap R(u)$ to the new children of $u$ and update the certificates of these points. Although splits can occur recursively, this can only happen if all points in $u$ must be reassigned to a single child $w$ of $u$. We therefore use the bounding box of points in $Q \cap R(u)$ to ensure that we only need to reassign points (and recompute certificates) once in a string of recursive splits. We can use a similar strategy for merges. 

Finally, we need to update which cells are local maxima, as this can change for each affected node $v \in \WT$ as well as their neighborhoods $\mathcal{N}(v)$. We would like to use Lemma~\ref{lem:cell-local-maxima} to bound the number of neighbors we need to consider, but that lemma holds only for a volume-based quadtree. To adapt the lemma to work with the weight-based quadtree, we first prove and adaptation of Lemma~\ref{lem:cell-local-maxima} for weight-based quadtrees.

\begin{lemma}\label{lem:weight-local-maxima}
Let $f = \KDE_P$ be a KDE function with Lipschitz constant $\lambda$ on a set of linearly moving points $P$ at a fixed time $t$, and let $Q$ be a coreset for $\KDE_P$ with additive error $\ec$. Furthermore, let $\WT$ be the weight-based quadtree on $Q$ with threshold $\rho \geq 8\ec$ at the same time $t$. If a cell $v \in \WT$ satisfies $\frac{\lambda s(v)^3}{\rho} \leq \frac{30}{\sqrt{2}}$, then for all $w \in \mathcal{N}(v)$ with $h(v) \geq h(w)$ it holds that $s(w) \geq \frac{1}{4} s(v)$.
\end{lemma}
\begin{proof}
In this proof we use $V(v)$ to refer to the actual volume under $f$ restricted to $R(v)$.

Let $v \in \WT$ be a leaf node with $\frac{\lambda s(v)^3}{\rho} \leq \frac{30}{\sqrt{2}}$ and let $w \in \mathcal{N}(v)$ be a neighboring cell of $v$ with $h(v) \geq h(w)$. We assume that $v$ is not at the deepest level of $\WT$ and that $s(w) < s(v)$, as otherwise the result is trivial. Hence, $W(v) \leq \rho$. We get that $h(w) \leq h(v) \leq \frac{\rho}{s(v)^2}$. By construction of $\WT$ we get that $W(w) \leq \frac{s(w)^2 \rho}{s(v)^2}$ and hence $V(w) \leq \frac{s(w)^2 \rho}{s(v)^2} + \ec$. In particular, there must be coordinates $(x, y) \in R(w)$ such that $f(x, y) \leq \frac{\rho}{s(v)^2} + \frac{\ec}{s(w)^2}$. Now consider the parent $u$ of $w$ in $\WT$. As $(x, y) \in R(u)$, we can apply Lemma~\ref{lem:vol-bounds}(c) to obtain that $V(u) \leq s(u)^2 (\frac{\rho}{s(v)^2} + \frac{\ec}{s(w)^2} + \frac{2 \sqrt{2}}{3} \lambda s(u))$. We now write $s(u) = \beta s(v)$ for some constant $\beta \in (0, 1]$. Since $u$ is not a leaf in $\WT$, we get that $V(u) \geq W(u) - \ec > \rho - \ec$. We thus obtain the following inequality: 
    \begin{alignat*}{2}
        && \beta^2 s(v)^2 (\frac{\rho}{s(v)^2} + \frac{4 \ec}{\beta^2 s(v)^2} + \frac{2 \sqrt{2}}{3} \lambda \beta s(v))
        &> \rho - \ec\\
        &\Rightarrow\qquad 
        & \beta^2 \rho + \frac{2 \sqrt{2}}{3} \lambda \beta^3 s(v)^3
        &> \rho - 5\ec\\
        &\Rightarrow\qquad 
        & \beta^2 \rho + \frac{2 \sqrt{2}}{3} \lambda \beta^3 s(v)^3
        &> \frac{3}{8} \rho\\        
        &\Rightarrow\qquad
        & \frac{2 \sqrt{2}}{3} \lambda \beta^3 s(v)^3
        &> (\frac{3}{8} - \beta^2)\rho \\
        &\Rightarrow\qquad
        & \frac{\lambda s(v)^3}{\rho}
        &> \frac{3}{2\sqrt{2}}\frac{\frac{3}{8} - \beta^2}{\beta^3} \\
    \end{alignat*}
We thus obtain that $\frac{3}{2\sqrt{2}}\frac{\frac{3}{8} - \beta^2}{\beta^3} < \frac{30}{\sqrt{2}}$ or $\frac{\frac{3}{8} - \beta^2}{\beta^3} < 20$. It is easy to verify that this inequality only holds for $\beta > \frac{1}{4}$. We can then directly conclude that $s(w) = \frac{1}{2}s(u) = \frac{1}{2} \beta s(v) > \frac{1}{8} s(v)$. As the ratios of sizes between cells in $T$ must always be a power of $2$, we conclude that $s(w) \geq \frac{1}{4} s(v)$.
\end{proof}

Now, we can use Lemma~\ref{lem:weight-local-maxima} to prove the following lemma:

\begin{restatable}{lemma}{localmaxima}\label{lem:local-maxima-bounded}
Let $f = \KDE_P$ be a KDE function with Lipschitz constant $\lambda$ and maximum $z^{*} \leq 1$ on a set of linearly moving points $P$ at a fixed time $t$. For any constant $0 < \varepsilon \leq 1$, let $Q$ be a coreset for $\KDE_P$ with additive error $\ec = \frac{\varepsilon^4 z^{*}}{48 \lambda^2 (8 + 2 \varepsilon z^{*})}$. Furthermore, let $\WT$ be the weight-based quadtree on $Q$ with threshold $\rho = \frac{\varepsilon^3}{6 \lambda^2 (8 + 2 \varepsilon z^{*})}$ at the same time $t$. If a cell $v \in \WT$ has $h(v) \geq \varepsilon$, then for all $w \in \mathcal{N}(v)$ with $h(v) \geq h(w)$ it holds that $s(w) \geq \frac{1}{4} s(v)$.
\end{restatable}

\begin{proof}
With the chosen values for $\rho$ and $\ec$ we have that $\rho = \frac{8 \ec}{z^{*} \varepsilon} \geq 8 \ec$ (for $\varepsilon, z^{*} \leq 1$). We can thus apply Lemma~\ref{lem:weight-local-maxima}. Now consider a cell $v \in \WT$ for which $\frac{\lambda s(v)^3}{\rho} > \frac{30}{\sqrt{2}}$ (otherwise we are done). Note that $\frac{\varepsilon^3}{60 \lambda^2} \leq \rho \leq \frac{\varepsilon^3}{48 \lambda^2}$. We thus get that $\frac{30}{\sqrt{2}} < \frac{\lambda s(v)^3}{\rho} < \frac{60 \lambda^3 s(v)^3}{\varepsilon^3}$ or $\frac{1}{\sqrt{2}} < \frac{\lambda s(v)}{\varepsilon}$, which implies that $s(v) > \frac{\varepsilon}{\sqrt{2} \lambda}$. Since $W(v) \leq \rho$ ($v$ is clearly not at the maximum depth of $\WT$), we get that $h(v) = \frac{W(v)}{s(v)^2} \leq \frac{\rho}{s(v)^2} \leq \frac{2 \lambda^2}{\varepsilon^2} \frac{\varepsilon^3}{48 \lambda^2} \leq \frac{1}{24} \varepsilon$. We thus conclude that the stated property must hold if $h(v) \geq \varepsilon$. 
\end{proof}

If a cell $v \in \WT$ is a local maximum of $f_{\WT}$ and has $h(v) < \varepsilon$, then any corresponding local maximum $(x, y)$ in $f = \KDE_P$ would have value $f(x, y) < 2 \varepsilon$ (since we have additive error $\varepsilon$), and hence this local maximum would not have persistence $\geq 2 \varepsilon$. Thus, by Lemma~\ref{lem:local-maxima-bounded}, a cell $v \in \WT$ can be a persistent local maximum only if its neighbors are of size at least $\frac{1}{4} s(v)$. Thus, to check if a node $v \in \WT$ is a local maximum, we can first check if $\SN(v)$ contains an internal (non-leaf) node $w$ with $s(w) = \frac{1}{4} s(v)$, in which case $v$ is not a persistent local maximum. Otherwise, we mark $v$ as a local maximum if $h(v) \geq h(w)$ for all $w \in \SN(v) \cap \mathcal{N}(v)$. To correctly keep track of local maxima, we need to perform this check on all nodes $v \in \WT$ affected by the event, as well as the nodes in $\SN(v) \cap \mathcal{N}(v)$ of an affected node $v$. 

\subsection{Analysis}

We analyze the kinetic data structure using the quality criteria as described in Section~\ref{sec:prelim}. However, as the KDS is built on a coreset $Q$ rather than the original point set $P$, this analysis is non-standard. This especially affects the efficiency of the KDS, as it is no longer clear what should be considered as an external event and what should be considered as the worst-case input (is the input $P$ or $Q$?). We therefore omit an analysis on the efficiency of this KDS and only consider the total number of events. On the other hand, we do include an analysis on flight-plan updates, as a change in trajectory of a point $p \in P$ requires the coreset $Q$ (the input of the KDS) to be updated, which is thus a very relevant operation to consider. 

\subparagraph{Responsiveness.} First, the new cell $u$ of a point $q \in Q$ can be computed in $O(d(\WT))$ time, where $d(\WT)$ is the depth of $\WT$. The number of nodes affected by an event is at most $O(d(\WT))$ as well (due to recursive splitting/merging). Since $|\SN(v)| = O(1)$ for all $v \in \WT$, we can check for all affected nodes (and neighbors) if they are a local maximum in $O(d(\WT))$ time in total. Similarly, we can update all pointers in $\SN(v)$ for all affected nodes $v \in \WT$ (and neighbors) in $O(d(\WT))$ time in total. By Corollary~\ref{cor:quaddepth} we get that $d(\WT) = O(\log\left(\frac{D}{\varepsilon}\right))$. Finally, we need to reassign points, recompute weights, and recompute certificates after splits/merges. We do this only once per event for all points in the corresponding cells. A single cell $u \in \WT$ can contain at most $(\rho + \ec) |Q| = O(\frac{1}{\varepsilon^5} \log(\frac{1}{\varepsilon}))$ points. Including updating the event queue, each point can be handled in $O(1) + O(\log |Q|)$ time. Thus, the total time to handle an event is $O(\frac{1}{\varepsilon^5} \log^2(\frac{1}{\varepsilon}) + \log\left(\frac{D}{\varepsilon}\right))$. 

\subparagraph{Locality and compactness.} Each point $q \in Q$ is involved in only $1$ certificate at a time, so the locality is $O(1)$ and the compactness is $O(|Q|) = O(\frac{1}{\varepsilon^8}\log(\frac{1}{\varepsilon}))$.

\subparagraph{Number of events.} Consider a grid over the domain $\dom = [0, D]^2$ with grid cells of size $\frac{1}{2} \sqrt{\rho} = O(\varepsilon^{\frac{3}{2}})$. As cells in $\WT$ cannot be smaller than $\frac{1}{2} \sqrt{\rho}$ by Lemma~\ref{lem:cellsize}, a point $q \in Q$ may trigger an event (that is, switch cells) at every grid line it crosses, in the worst case. Since $q$ follows linear motion, this means that $q$ cannot trigger more than $O(\frac{D}{\varepsilon^{\frac{3}{2}}})$ events. This results in a total number of events of $O(\frac{|Q| D}{\varepsilon^{\frac{3}{2}}}) = O(D\poly(\frac{1}{\varepsilon}))$. 

\subparagraph{Flight plan updates.}

We use the dynamic data structure by Agarwal~\etal~\cite{approx-rangespace} to facilitate maintenance of our coreset during flight plan updates. We briefly review how this dynamic data structure works. 

The data structure maintains a set of $O(\log n)$ perfectly balanced trees of different ranks, where a tree of rank $i$ contains $2^i$ points. In these trees, each internal node stores an $\varepsilon$-approximation of all the points stored in the subtree of the node. To insert a new point, we make a new tree of rank $0$ with only that point in it. To remove a point, we break the tree containing the point into $O(\log n)$ trees. Afterwards, we merge trees of the same rank $i$ into a new tree of rank $i+1$, and we perform merging and halving steps (as described above) to obtain an $\varepsilon$-approximation for the new root. We repeat merging trees until every tree has a unique rank. Finally, the $\varepsilon$-approximation of the full set of points can be obtained by computing an $\varepsilon$-approximation of the union of $\varepsilon$-approximations of all tree roots. This result is summarized in Theorem~\ref{thm:eps-approx-rangespace}.

\dynamicapprox*

In our setting, a single point $p \in P$ corresponds to multiple points in $Q$, and specifically also in the ground set $\mathcal{S}$. Thus, we would have to delete and insert multiple points to handle the change of trajectory of a single point. However, we can see the $\varepsilon$-approximation of the volume under the kernel $K$ as a single element in the dynamic data structure in Theorem~\ref{thm:eps-approx-rangespace}, as we can simply copy it for every point $p \in P$. Thus, we can use the data structure with $n$ being the number of points in $P$ rather than being the number of points in $\mathcal{S}$, and every flight plan update to a single point $p \in P$ can be handled by a single insertion and deletion in the data structure. We then obtain a new coreset $Q'$ that we can use to compute an updated weight-based quadtree. Using this data structure, we can compute a new coreset $Q'$ in $O\left(\poly\left(\frac{\log n}{\varepsilon}\right)\right)$ time. Afterwards, we must update $\WT$ (and all associated certificates) to use $Q'$ instead of $Q$. In the worst case we may need to rebuild $\WT$ completely (as $Q'$ and $Q$ might be completely different), but the construction time of $\WT$ is bounded by $O(|\WT| + |Q|) = O\left(\poly\left(\frac{\log n}{\varepsilon}\right)\right)$.     

\section{Discussion}\label{sec:discussion}

We presented a KDS that efficiently tracks persistent local maxima of a KDE on a set of linearly moving points $P$. To develop this KDS, we first showed how to approximate (within a given error bound) a density function via a volume-based quadtree. We then proved that we can compute a coreset of moving points which approximates the volume under a density function. A weight-based quadtree on this coreset in turn approximates the volume-based quadtree on the density function. 
For any $\varepsilon > 0$, we can compute this coreset of size $O(\poly(\frac{1}{\varepsilon}))$ in $O\left(n \poly\left(\frac{\log n}{\varepsilon}\right)\right)$ time, where $n$ is the number of points in $P$ and $\varepsilon$ is the error bound between the weight-based quadtree and the density function. 

Various bounds on the quadtree complexity and the KDS quality measures depend on the size of the domain $D$. As we assume that our input points represent a single group, it makes sense to assume that the kernel functions of any point (its region of influence) must overlap with the kernel function of at least one other point. Since we scale the input such that the kernel width is $\sigma = 1$, this directly implies that $D = O(n)$ for a static set of points, although it is likely much smaller. However, when points move in a single direction for a long time (say, when a herd is migrating), they may easily leave a domain of that size. To address this problem without blowing up the size of the domain, we can move the domain itself along a piecewise-linear trajectory. A change of direction of the domain directly changes the trajectories of all points, and  all events in the KDS must be recomputed. The coreset, however, does not need to change during such an event. We can limit the number of domain flight plan changes by using a slightly larger domain than needed at any point in time.

We believe that approximating the density surface we want to maintain via a suitable coreset of moving points is a promising direction also in practice. Below we briefly sketch the necessary adaptations that we foresee. 
For the coreset to exist, the range space formed by the trajectories of the (samples around) the input objects and a set of square regions needs to have bounded VC-dimension. We proved an upper bound on this VC-dimension in the case that all trajectories are linear. Real-world animal trajectories are certainly not linear. However, since animals cannot move at arbitrary speeds and subgroups can often be observed to stay together, we still expect the corresponding range space to have bounded VC-dimension. A formal proof seems out of reach for more than very restrictive motion models, but bounds might be deduced from experimental data. 

The actual computation of the coreset via the algorithm of Agarwal~\etal~\cite{approx-rangespace} is impossible if the trajectories are not known ahead of time. However, random sampling (that is, sample a point $p \in P$, and then sample from its kernel) can be expected to result in a coreset of good size and quality in practice (the worst-case bounds on the size of the coreset which we proved are unlikely to be necessary in practice). Since we generally do not know the trajectories of the animals, but we do have bounds on their maximum speeds, a black-box KDS could be used to maintain such a random sampling coreset efficiently.

Our theoretical results inform the direction of our future engineering efforts in two ways. First of all, we now know that we can approximate well with a coreset whose size depends only on the desired approximation factor and not on the input size. Second, we know how to sample to find such a coreset, by essentially constructing randomly shifted copies of the input points.



\bibliography{kineticKDE}

\appendix

\end{document}